\documentclass[11pt,letterpaper]{article}

\usepackage{amsmath,amsfonts,amsthm,amssymb,stmaryrd,relsize}
\theoremstyle{plain}

\numberwithin{equation}{section}

\newtheorem{thm}{Theorem}[section]
\newtheorem{lem}[thm]{Lemma}
\newtheorem{cor}[thm]{Corollary}

\newenvironment{exam}[1]
{\begin{flushleft}\textbf{Example #1}.\enspace}%
{\end{flushleft}}

\allowdisplaybreaks  

\newcommand{\real}{{\mathbb R}}
\newcommand{\rmtr}{\mathrm{tr\,}}

\newcommand{\dscript}{\mathcal{D}}
\newcommand{\escript}{\mathcal{E}}
\newcommand{\hscript}{\mathcal{H}}
\newcommand{\iscript}{\mathcal{I}}
\newcommand{\jscript}{\mathcal{J}}

\newcommand{\lscript}{\mathcal{L}}
\newcommand{\mscript}{\mathcal{M}}
\newcommand{\sscript}{\mathcal{S}}

\newcommand{\brac}[1]{\left\{#1\right\}}
\newcommand{\paren}[1]{\left(#1\right)}
\newcommand{\sqbrac}[1]{\left[#1\right]}
\newcommand{\elbows}[1]{{\left\langle#1\right\rangle}}

\begin{document}

\title{ENTROPY OF QUANTUM MEASUREMENTS}
\author{Stan Gudder\\ Department of Mathematics\\
University of Denver\\ Denver, Colorado 80208\\
sgudder@du.edu}
\date{}
\maketitle

\begin{abstract}
If $a$ is a quantum effect and $\rho$ is a state, we define the $\rho$-entropy $S_a(\rho )$ which gives the amount of uncertainty that a measurement of $a$ provides about $\rho$. The smaller $S_a(\rho )$ is, the more information a measurement of $a$ gives about $\rho$. In Section~2, we provide bounds on $S_a(\rho )$ and show that if $a+b$ is an effect, then $S_{a+b}(\rho )\ge S_a(\rho )+S_b(\rho )$. We then prove a result concerning convex mixtures of effects. We also consider sequential products of effects and their $\rho$-entropies. In Section~3, we employ $S_a(\rho )$ to define the $\rho$-entropy $S_A(\rho )$ for an observable $A$. We show that $S_A(\rho )$ directly provides the $\rho$-entropy $S_\iscript (\rho )$ for an instrument $\iscript$. We establish bounds for $S_A(\rho )$ and prove characterizations for when these bounds are obtained. These give simplified proofs of results given in the literature. We also consider $\rho$-entropies for measurement models, sequential products of observables and coarse-graining of observables. Various examples that illustrate the theory are provided.
\end{abstract}

\section{Introduction}  
In an interesting article, D.~\v{S}afr\'anek and J.~Thingna introduce the concept of entropy for quantum instruments \cite{st22}. Various important theorems are proved and applications are given. In quantum computation and information theory one of the most important problems is to determine an unknown state by applying measurements on the system \cite{hz12,nc00,op04,sasd21}. Entropy provides a quantification for the amount of information given to solve this so-called state discrimination problem \cite{lin75,von55,weh78}. In this article, we first define the entropy for the most basic measurement, namely a quantum effect $a$ \cite{blm96,hz12,kra83,nc00}. If $\rho$ is a state, we define the $\rho$-entropy $S_a(\rho )$ which gives the amount of uncertainty (or randomness) that a measurement of $a$ provides about $\rho$. The smaller $S_a(\rho )$ is, the more information a measurement of $a$ provides about $\rho$. In Section~2, we give bounds on $S_a(\rho )$ and show that if $a+b$ is an effect then $S_{a+b}(\rho )\le S_a(\rho )+S_b(\rho )$. We then prove a result concerning convex mixtures of effects. We also consider sequential products of effects and their $\rho$-entropies.

In Section~3, we employ $S_a(\rho )$ to define the entropy $S_A(\rho )$ for an observable $A$. Then $S_A(\rho )$ gives the uncertainty that a measurement of $A$ provides about $\rho$. We show that $S_A(\rho )$ directly gives the $\rho$-entropy $S_\iscript (\rho )$ for an instrument
$\iscript$. We establish bounds for $S_A(\rho )$ and characterize when these bounds are obtained. These give simplified proofs of results given in \cite{sda19,sasd21,st22}. We also consider $\rho$-entropies for measurement models, sequential products of observables and coarse-graining of observables. Various examples that illustrate the theory are provided. In this work, all Hilbert spaces are assumed to be finite dimensional. Although this is a restriction, the work applies for quantum computation and information theory \cite{blm96,hz12,kra83,nc00}.

\section{Entropy for Effects}  
Let $H$ be a finite dimensional complex Hilbert space with dimension $n$. We denote the set of linear operators on $H$ by $\lscript (H)$ and the set of states on $H$ by $\sscript (H)$. If $\rho\in\sscript (H)$ with nonzero eigenvalues $\lambda _1,\lambda _2,\ldots ,\lambda _m$ including multiplicities, the
\textit{von Neumann entropy} of $\rho$ is \cite{lin75,op04,von55,weh78}.
\begin{equation*}
S(\rho )=-\sum _{i=1}^m\lambda _i\ln (\lambda _i)=-\rmtr\sqbrac{\rho\ln (\rho )}
\end{equation*}
We consider $S(\rho )$ as a measure of the randomness or uncertainty of $\rho$ and smaller values of $S(\rho )$ indicate more information content. For example, $\rho$ is the completely random state $I/n$, where $I$ is the identity operator, if and only if $S(\rho )=\ln (n)$ and $\rho$ is a pure state if and only if $S(\rho )=0$. Moreover, it is well-known that $0\le S(\rho )\le\ln (n)$ for all $\rho\in\sscript (H)$. The following properties of $S$ are well-known \cite{lin75,op04,weh78}:
\begin{align*}
S(U\rho U^*)&=S(\rho )\hbox{ when $U$ is unitary}\\
S(\rho _1\otimes\rho _2)&=S(\rho _1)+S(\rho _2)\\
\sum\mu _iS(\rho _i)&\le S\paren{\sum\mu _i\rho _i}\le\sum\mu _iS(\rho _i)-\sum\mu _i\ln (\mu _i)
\end{align*}
where $0\le\mu _i=1$ with $\sum\mu _i=1$.

An operator $a\in\lscript (H)$ that satisfies $0\le a\le I$ is called an \textit{effect} \cite{blm96,hz12,kra83,nc00}. We think of an effect $a$ as a two-outcome yes-no measurement. If a measurement of $a$ results in outcome yes we say that $a$ \textit{occurs} and if it results in outcome no then $a$ \textit{does not occur}. The effect $a'=I-a$ is the \textit{complement} of $a$ and $a'$ occurs if and only if $a$ does not occur. We denote the set of effects by $\escript (H)$. If $a\in\escript (H)$ and $\rho\in\sscript (H)$ then $0\le\rmtr (\rho a)\le 1$ and we interpret
$\rmtr (\rho a)$ as the probability that $a$ occurs when the system is in state $\rho$. If $a\ne 0$ we define the $\rho$-\textit{entropy} of $a$ to be
\begin{equation}                
\label{eq21}
S_a(\rho )=-\rmtr (\rho a)\ln\sqbrac{\tfrac{\rmtr (\rho a)}{\rmtr (a)}}
\end{equation}
We interpret $S_a (\rho )$ as the amount of uncertainty that the system is in state $\rho$ resulting from a measurement of $a$. The smaller $S_a(\rho )$ is, the more information a measurement of $a$ gives about $\rho$. Such information is useful for state discrimination problems \cite{hz12,nc00,op04,sasd21}.

If $\rho$ is the completely random state $I/n$ then \eqref{eq21} becomes
\begin{equation*}
S_a (I/n)=-\rmtr (Ia/n)\ln\sqbrac{\tfrac{\rmtr (Ia/n)}{\rmtr (a)}}=-\tfrac{1}{n}\,\rmtr (a)\ln\paren{\tfrac{1}{n}}=\tfrac{\rmtr (a)}{n}\,\ln (n)
\end{equation*}
Since $\rmtr (a)\le n$ we conclude that $S_a(I/n)\le S(I/n)$ for all $a\in\escript (H)$. Another extreme case is when $a=\lambda I$ for
$0<\lambda\le 1$. We then have for any $\rho\in\sscript (H)$ that
\begin{equation*}
S_{\lambda I}(\rho )=-\rmtr (\rho\lambda I)\ln\sqbrac{\tfrac{\rmtr (\rho\lambda I)}{\rmtr (\lambda I)}}
  =-\lambda\ln\sqbrac{\tfrac{\lambda}{\lambda\rmtr (I)}}=\lambda\ln (n)
\end{equation*}
Thus, as $\lambda$ gets smaller, the more information we gain.

A real-valued function with domain $\dscript (f)$, an interval in $\real$, is \textit{strictly convex} if for any $x_1,x_2\in\dscript (f)$ with $x_1\ne x_2$ and $0<\lambda <1$ we have
\begin{equation*}
f\sqbrac{\lambda x_1+(1-\lambda )x_2}<\lambda f(x_1)+(1-\lambda )f(x_2)
\end{equation*}
If the opposite inequality holds, then $f$ is \textit{strictly concave}. It is clear that $f$ is strictly convex if and only if $-f$ is strictly concave. Of special importance in this work are the strictly convex functions $-\ln x$ and $x\ln x$. We shall frequently employ Jensen's theorem which says: if $f$ is strictly convex and $0\le\mu _i\le 1$ with $\sum\limits _{i=1}^m\mu _i=1$, then
\begin{equation*}
f\paren{\sum _{i=1}^m\mu _ix_i}\le\sum _{i=1}^m\mu _if(x_i)
\end{equation*}
Moreover, we have equality if and only if $x_i=x_j$ for all $i,j=1,2,\ldots ,m$ \cite{st22}.

\begin{thm}    
\label{thm21}
If $\rho\in\sscript (H)$ with nonzero eigenvalues $\lambda _i$, $i=1,2,\ldots, m$, and $a\in\escript (H)$ with $\rmtr (\rho a)\ne 0$, then
\begin{equation*}
-\sum _i\rmtr (P_ia)\lambda _i\ln (\lambda _i)\le S_a(\rho )\le\ln\sqbrac{\tfrac{\rmtr (a)}{\rmtr (\rho a)}}
\end{equation*}
where $\rho =\sum _i\lambda _iP_i$ is the spectral decomposition of $\rho$. Moreover, $S_a(\rho )=\ln\sqbrac{\rmtr (a)/\rmtr (\rho a)}$ if and only if 
$\rmtr (\rho a)=1$ in which case $S_a(\rho )=ln\sqbrac{\rmtr (a)}$ and if
\begin{equation}                
\label{eq22}
S_a(\rho )=-\sum _i\rmtr (P_ia)\lambda _i\ln (\lambda _i)
\end{equation}
then $\rmtr (P_ia)=\rmtr (P_ja)$ for all $i,j =1,2,\ldots ,m$ and $S_a(\rho )=(\rmtr (a)/m) S(\rho )$ while if $\rmtr (P_ia)=\rmtr (P_ja)$ for all $i,j=1,2,\ldots m$ then $S_a(\rho )=(\rmtr (a)/m)\ln (m)$.
\end{thm}
\begin{proof}
Letting $\mu _j=\rmtr (P_ja)/\rmtr (a)$, $j=1,2,\ldots ,m$, we have that $0\le\mu _j\le 1$ and $\sum _j\mu _j=1$. Since $-x\ln (x)$ is strictly concave we obtain
\begin{align*}
S_a(\rho )&=-\rmtr (\rho a)\ln\sqbrac{\tfrac{\rmtr (\rho a)}{\rmtr (a)}}
   =-\rmtr\paren{\sum _i\lambda _iP_ia}\ln\sqbrac{\tfrac{\rmtr\paren{\sum _j\lambda _jP_ja}}{\rmtr (a)}}\\
   &=-\sum\lambda _i\rmtr (P_ia)\ln\paren{\sum _j\lambda _j\mu _j}=\rmtr (a)\sqbrac{-\sum _i\lambda _i\mu _i\paren{\sum _j\lambda _j\mu _j}}\\
   &\ge -\rmtr (a)\sum _i\mu _i\lambda _i\ln (\lambda _i)=-\rmtr (a)\sum _i\tfrac{\rmtr (P_ia)}{\rmtr (a)}\,\lambda _i\ln (\lambda _i)\\
   &=-\sum _i\rmtr (P_ia)\lambda _i\ln (\lambda _i)
\end{align*}
Since
\begin{equation*}
\rmtr (\rho a)=\rmtr (a^{1/2}\rho a^{1/2})\le\rmtr (\rho )=1
\end{equation*}
we have that
\begin{equation*}
S_a(\rho )=\rmtr (\rho a)\ln\sqbrac{\tfrac{\rmtr (a)}{\rmtr (\rho a)}}\le\ln\sqbrac{\tfrac{\rmtr (a)}{\rmtr (\rho a)}}
\end{equation*}
If $\rmtr (\rho a)=1$, then
\begin{equation*}
S_a(\rho )=-\rmtr (\rho a)\ln\sqbrac{\tfrac{\rmtr (\rho a)}{\rmtr (\rho a)}}=-\ln\sqbrac{\tfrac{1}{\rmtr (a)}}=\ln\sqbrac{\rmtr (a)}
\end{equation*}
Conversely, if $S_a(\rho )=\ln\sqbrac{\rmtr (a)/\rmtr (\rho a)}$, then clearly $\rmtr (\rho a)=1$. If \eqref{eq22} holds, then we have equality for Jensen's inequality. Hence, $\rmtr (P_ia)=\rmtr (P_ja)$ for all $i,j=1,2,\ldots ,m$. Since
\begin{equation*}
\rmtr (a)=\sum _i\rmtr (P_ia)=m\rmtr (P_ia)
\end{equation*}
we conclude that
\begin{equation*}
S_a(\rho )=-\rmtr (P_1a)\sum _i\lambda _i\ln (\lambda _i)=\tfrac{\rmtr (a)}{m}\,S(\rho )
\end{equation*}
Finally, suppose $\rmtr (P_ia)=\rmtr (P_ja)$ for all $i,j=1,2,\ldots ,m$. Then
\begin{equation*}
\rmtr (a)=\sum _i\rmtr (P_ia)=m\rmtr (P_1a)
\end{equation*}
We conclude that
\begin{align*}
S_a(\rho )&=-\rmtr (P_1a)\sum _i\lambda _i\ln\sqbrac{\sum _j\lambda _j\tfrac{\rmtr (P_1a)}{\rmtr (a)}}
   =-\rmtr (P_1a)\sum _i\lambda _i\ln\paren{\sum _j\lambda _j\tfrac{1}{m}}\\
   &=-\rmtr (P_1a)\sum _i\lambda _i\ln\paren{\tfrac{1}{m}}=\tfrac{\rmtr (a)}{m}\,\ln (m)\qedhere
\end{align*}
\end{proof}

For $a,b\in\escript (H)$ we write $a\perp b$ if $a+b\in\escript (H)$.

\begin{thm}    
\label{thm22}
If $a\perp b$, then $S_{a+b}(\rho )\ge S_a(\rho )+S_b(\rho )$ for all $\rho\in\sscript (H)$. Moreover, $S_{a+b}(\rho )=S_a(\rho )+S_b(\rho )$ if and only if $\rmtr (b)\rmtr (\rho a)=\rmtr (a)\rmtr (\rho b)$.
\end{thm}
\begin{proof}
Since $-x\ln x$ is concave, letting $\lambda _1=\rmtr (a)/\sqbrac{\rmtr (a)+\rmtr (b)}$, $\lambda _2=\rmtr (b)/\sqbrac{\rmtr (a)+\rmtr (b)}$,
$x_1=\rmtr (\rho a)/\rmtr (a)$, $x_2=\rmtr (\rho b)/\rmtr (b)$ we obtain
\begin{align*}
S_{a+b}(\rho )&=-\rmtr\sqbrac{\rho (a+b)}\ln\brac{\tfrac{\rmtr\sqbrac{\rho (a+b)}}{\rmtr (a+b)}}\\
  &=-\rmtr (a+b)\sqbrac{\tfrac{\rmtr (\rho a)+\rmtr (\rho b)}{\rmtr (a+b)}}\ln\sqbrac{\tfrac{\rmtr (\rho a)+\rmtr (\rho b)}{\rmtr (a+b)}}\\
  &=-\rmtr (a+b)(\lambda _1x_1+\lambda _2x_2)\ln (\lambda _1x_1+\lambda _2x_2)\\
  &\ge -\rmtr (a+b)\sqbrac{\lambda _1x_1\ln (x_1)+\lambda _2x_2\ln (x_2)}\\
  &=-\rmtr (\rho a)\ln\sqbrac{\tfrac{\rmtr (\rho a)}{\rmtr (a)}}-\rmtr (\rho b)\ln\sqbrac{\tfrac{\rmtr (\rho b)}{\rmtr (b)}}=S_a(\rho )+S_b(\rho )
\end{align*}
We have equality if and only if $x_1=x_2$ which is equivalent to $\rmtr (b)\rmtr (\rho a)=\rmtr (a)\rmtr (\rho b)$.
\end{proof}

\begin{cor}    
\label{cor23}
$S_a(\rho )+S_{a'}(\rho )\le\ln (n)$ and $S_a(\rho )+S_{a'}(\rho )=\ln (n)$ if and only if $\rmtr (a)=n\rmtr (\rho a)$.
\end{cor}
\begin{proof}
Applying Theorem~\ref{thm22} we obtain
\begin{equation*}
S_a(\rho )+S_{a'}(\rho )\le S_{a+a'}(\rho )=S_I(\rho )=\ln (n)
\end{equation*}
\begin{align*}
\hbox{We have equality }&\Leftrightarrow\rmtr (a')\rmtr (\rho a)=\rmtr (a)\rmtr (\rho a')\\
   &\Leftrightarrow\sqbrac{n-\rmtr (a)}\rmtr (\rho a)=\rmtr (a)\sqbrac{1-\rmtr (\rho a)}\\
   &\Leftrightarrow\rmtr (a)=n\rmtr (\rho a)\qedhere
\end{align*}
\end{proof}

\begin{cor}    
\label{cor24}
$S_{a+b}(\rho )\ge S_a(\rho ),S_b(\rho)$.
\end{cor}

\begin{cor}    
\label{cor25}
If $a\le b$, then $S_a(\rho )\le S_b(\rho )$ for all $\rho\in\sscript (H)$.
\end{cor}
\begin{proof}
If $a\le b$, then $b=a+c$ for $c=b-a\in\escript (H)$. Hence,
\begin{equation*}
S_b(\rho )=S_{a+c}(\rho )\ge S_a(\rho )+S_c(\rho )\ge S_a(\rho )
\end{equation*}
for every $\rho\in\sscript (H)$.
\end{proof}

Applying Theorem~\ref{thm22} and induction we obtain the following.

\begin{cor}    
\label{cor26}
If $a_1+a_2+\cdots +a_m\le I$, then $S_{\sum a_i}(\rho )\ge\sum S_{a_i}(\rho )$. Moreover, we have equality if and only if
$\rmtr (a_j)\rmtr (\rho a_i)=\rmtr (a_i)\rmtr (\rho a_j)$ for all $i,j=1,2,\ldots ,m$.
\end{cor}

Notice that $\escript (H)$ is a convex set in the sense that if $a_i\in\escript (H)$ and $0\le\lambda _i\le 1$ with
$\sum _{i=1}^m\lambda _i=1$, then $\sum\lambda _ia_i\in\escript (H)$.

\begin{cor}    
\label{cor27}
{\rm{(i)}}\enspace If $0<\lambda\le 1$ and $a\in\escript (H)$, then $S_{\lambda a}(\rho )=\lambda S_a(\rho )$ for all $\rho\in\sscript (H)$.
{\rm{(ii)}}\enspace If $0<\lambda _i\le 1$, $a_i\in\escript (H)$, with $\sum\limits _{i=1}^m\lambda _i=1$, then
$S_{\sum\lambda _ia_i}(\rho )\le\sum\lambda _iS_{a_i}(\rho )$ for all $\rho\in\sscript (H)$. We have equality if and only if
$\rmtr (a_j)\rmtr (\rho a_i)=\rmtr (a_i)\rmtr (\rho a_j)$ for all $i,j =1,2,\ldots ,m$.
\end{cor}
\begin{proof}
(i)\enspace We have that
\begin{equation*}
S_{\lambda a}(\rho )=-\rmtr (\rho\lambda a)\ln\sqbrac{\tfrac{\rmtr (\rho\lambda a)}{\rmtr (\lambda a)}}
   =-\rmtr (\rho a)\ln\sqbrac{\tfrac{\lambda\rmtr (\rho a)}{\lambda\rmtr (a)}}=\lambda S_a(\rho )
\end{equation*}
(ii)\enspace Applying (i) and Corollary~\ref{cor26} gives
\begin{equation*}
S_{\sum\lambda _ia_i}(\rho )\ge\sum S_{\lambda _ia_i}(\rho )=\sum\lambda _iS_{a_i}(\rho )
\end{equation*}
together with the equality condition.
\end{proof}

As with $\escript (H)$, $\sscript (H)$ is a convex set and we have the following.

\begin{thm}    
\label{thm28}
If $0<\lambda _i\le 1$ $\rho _i\in\sscript (H)$, $i=1,2,\ldots ,m$, with $\sum\limits _{i=1}^m\lambda _i=1$, then
\begin{equation*}
S_a\paren{\sum\lambda _i\rho _i}\ge\sum\lambda _iS_a(\rho _i)
\end{equation*}
for all $a\in\escript (H)$. We have equality if and only if $\rmtr (\rho _ia)=\rmtr (\rho _ja)$ for all $i,j=1,2,\ldots ,m$.
\end{thm}
\begin{proof}
Letting $x_i=\rmtr (\rho _ia)/\rmtr (a)$, since $-x\ln x$ is concave, we obtain
\begin{align*}
S_a\paren{\sum\lambda _i\rho _i}&=-\rmtr\paren{\sum\lambda _i\rho _ia}\ln\sqbrac{\tfrac{\rmtr\paren{\sum\lambda _i\rho _ia}}{\rmtr (a)}}\\
   &=-\rmtr (a)\sum\lambda _i\tfrac{\rmtr (\rho _ia)}{\rmtr (a)}\ln\sqbrac{\tfrac{\sum\lambda _i\rmtr (\rho _ia)}{\rmtr (a)}}\\
   &=\rmtr (a)\sqbrac{-\sum\lambda _ix_i\ln\paren{\sum\lambda _jx_j}}\ge -\rmtr (a)\sum\lambda _ix_i\ln (x_i)\\
   &=-\rmtr (a)\sum\lambda _i\tfrac{\rmtr (\rho _ia)}{\rmtr (a)}\ln\sqbrac{\tfrac{\rmtr (\rho _ia)}{\rmtr (a)}}
   =-\sum\lambda _i\rmtr (\rho _ia)\ln\sqbrac{\tfrac{\rmtr (\rho _ia)}{\rmtr (a)}}\\
   &=\sum\lambda _iS_a(\rho _i)
\end{align*}

We have equality if and only if $x_i=x_j$ which is equivalent to $\rmtr (\rho _ia)=\rmtr (\rho _ja)$ for all $i,j=1,2,\ldots ,m$.
\end{proof}

\begin{thm}    
\label{thm29}
If $a_i\in\escript (H_i)$, $\rho _i\in\sscript (H_i)$, $i=1,2$, then
\begin{equation*}
S_{a_1\otimes a_2}(\rho _1\otimes\rho _2)=\rmtr (\rho _2a_2)S_{a_1}(\rho _1)+\rmtr (\rho _1a_1)S_{a_2}(\rho _2)
   \le S_{a_1}(\rho _1)+S_{a_2}(\rho _2)
\end{equation*}
\end{thm}
\begin{proof}
This follows from
\begin{align*}
S_{a_1\otimes a_2}(\rho _1\otimes\rho _2)
  &=-\rmtr (\rho _1\otimes\rho _2a_1\otimes a_2)\ln\sqbrac{\tfrac{\rmtr (\rho _1\otimes\rho _2a_1\otimes a_2)}{\rmtr (a_1\otimes a_2)}}\\
  &=-\rmtr (\rho _1a_1)\rmtr (\rho _2a_2)\ln\sqbrac{\tfrac{\rmtr (\rho _1a_1)\rmtr (\rho _2a_2)}{\rmtr (a_1)\rmtr (a_2)}}\\
  &=-\rmtr (\rho _1a_1)\rmtr (\rho _2a_2)
  \brac{\ln\sqbrac{\tfrac{\rmtr (\rho _1a_1)}{\rmtr (a_1)}}+\ln\sqbrac{\tfrac{\rmtr (\rho _2a_2)}{\rmtr (a_2)}}}\\
  &=\rmtr (\rho _2a_2)S_{a_1}(\rho _1)+\rmtr (\rho _1a_1)S_{a_2}(\rho _2)\le S_{a_1}(\rho _1)+S_{a_2}(\rho _2)\qedhere
\end{align*}
\end{proof}

An \textit{operation} on $H$ is a completely positive linear map $\iscript\colon\lscript (H)\to\lscript (H)$ such that $\rmtr\sqbrac{\iscript (A)}\le\rmtr (A)$ for all $A\in\lscript (H)$ \cite{blm96,hz12,kra83,lin75,nc00}. If $\iscript$ is an operation we define the \textit{dual} of $\iscript$ to be the unique linear map $\iscript ^*\colon\lscript (H)\to\lscript (H)$ that satisfies $\rmtr\sqbrac{\iscript (A)B}=\rmtr\sqbrac{A\iscript ^*(B)}$ for all
$A,B\in\lscript (H)$. If $a\in\escript (H)$ then for any $\rho\in\sscript (H)$ we have $0\le\rmtr\sqbrac{\iscript (\rho )a}\le 1$ and it follows that
$\iscript ^*(a)\in\escript (H)$. We say that $\iscript$ \textit{measures} $a\in\escript (H)$ if $\rmtr\sqbrac{\iscript (\rho )}=\rmtr (\rho a)$ for all
$\rho\in\sscript (H)$. If $\iscript$ measures $a$ we define the $\iscript$-\textit{sequential product} $a\circ b=\iscript ^*(b)$ for all
$b\in\escript (H)$ \cite{gud20,gud21}. Although $a\circ b$ depends on the operation used to measure $a$ we do not include $\iscript$ in the notation for simplicity. We interpret $a\circ b$ as the effect that results from first measuring $a$ using $\iscript$ and then measuring $b$.

\begin{thm}    
\label{thm210}
{\rm{(i)}}\enspace If $b\perp c$, then $a\circ (b+c)=a\circ b+a\circ c$.
{\rm{(ii)}}\enspace $a\circ I=a$.
{\rm{(iii)}}\enspace $a\circ b\le a$ for all $b\in\escript (H)$.
{\rm{(iv)}}\enspace $S_{a\circ b}(\rho )\le S_a(\rho )$ for all $\rho\in\sscript (H)$.
\end{thm}
\begin{proof}
(i)\enspace For every $\rho\in\sscript (H)$ we obtain
\begin{align*}
\rmtr\sqbrac{\rho\,a\circ (b+c)}&=\rmtr\sqbrac{\rho\iscript ^*(b+c)}=\rmtr\sqbrac{\iscript (\rho )(b+c)}
  =\rmtr\sqbrac{\iscript (\rho )b}+\rmtr\sqbrac{\iscript (\rho )c}\\
  &=\rmtr\sqbrac{\rho\iscript ^*(b)}+\rmtr\sqbrac{\rho\iscript ^*(c)}=\rmtr\sqbrac{\rho\,a\circ b}+\rmtr\sqbrac{\rho\,a\circ c}\\
  &=\rmtr\sqbrac{\rho (a\circ b+a\circ c)}
\end{align*}
Hence, $a\circ (b+c)=a\circ b+a\circ c$.
(ii)\enspace For all $\rho\in\sscript (H)$ we have
\begin{equation*}
\rmtr (\rho\,a\circ I)=\rmtr\sqbrac{\rho\iscript ^* (I)}=\rmtr\sqbrac{\iscript (\rho )I}=\rmtr\sqbrac{\iscript (\rho )}=\rmtr (\rho a)
\end{equation*}
Hence, $a\circ I=a$.
(iii)\enspace By (i) and (ii) we have
\begin{equation*}
a\circ b+a\circ b'=a\circ (b+b')=a\circ I=a
\end{equation*}
It follows that $a\circ b\le a$.
(iv)\enspace Since $a\circ b\le a$, by Corollary~\ref{cor25} we obtain $S_{a\circ b}(\rho )\le S_a(\rho )$ for all $\rho\in\sscript (H)$.
\end{proof}

Theorem~\ref{thm210}(iv) shows that $a\circ b$ gives more information than $a$ about $\rho$. We can continue this process and make more measurements as follows. If $\iscript ^i$ measures $a^i$, $i=1,2,\ldots ,m$, we have
\begin{equation*}
a^1\circ a^2\circ\cdots\circ a^m=(\iscript ^1)^*(\iscript ^2)^*\cdots (\iscript ^{m-1})^*(a^m)
\end{equation*}
and it follows from Theorem~\ref{thm210}(iv) that
\begin{equation*}
S_{a^1\circ a^2\circ\cdots\circ a^m}(\rho )\le S_{a^1\circ a^2\circ\cdots\circ a^{m-1}}(\rho )
\end{equation*}
Notice that the probability of occurrence of the effect $a^1\circ a^2\circ\cdot\circ a^m$ in state $\rho$ is
\begin{align*}
\rmtr (\rho\,a^1\circ a^2\circ\cdots\circ a^m)&=\rmtr\sqbrac{\rho (\iscript ^1)^*(\iscript ^2)^*\cdots (\iscript ^{m-1})^*(a^m)}\\
   &=\rmtr\sqbrac{\iscript ^{m-1}\iscript ^{m-2}\cdots\iscript ^1(\rho )a^m}
\end{align*}
Thus, we begin with the input state $\rho$, then measure $a^1$ using $\iscript ^1$, then measure $a^2$ using $\iscript ^2,\ldots$ and finally measuring $a^m$.

\begin{exam}{1}  
For $a\in\escript (H)$ we define the \textit{L\"uders operation} $\lscript ^a(A)=a^{1/2}Aa^{1/2}$ \cite{lud51}. Since
\begin{equation*}
\rmtr\sqbrac{A(\lscript ^a)^*(B)}=\sqbrac{\lscript ^a(A)B}=\rmtr\sqbrac{a^{1/2}Aa^{1/2}B}=\rmtr (Aa^{1/2}Ba^{1/2})
\end{equation*}
we have $(\lscript ^a)^*(B)=a^{1/2}Ba^{1/2}$ so $(\lscript ^a)^*=\lscript ^a$. We have that $\lscript ^a$ measures $a$ because
\begin{equation*}
\rmtr\sqbrac{\lscript ^a(\rho )}=\rmtr (a^{1/2}\rho a^{1/2})=\rmtr (\rho a)
\end{equation*}
for every $\rho\in\sscript (H)$. We conclude that the $\lscript ^a$ sequential product is
\begin{equation*}
a\circ b=(\lscript ^a)^*(b)=a^{1/2}ba^{1/2}
\end{equation*}
We also have that
\begin{align*}
S_{a\circ b}(\rho )&=-\rmtr (\rho\,a\circ b)\ln\sqbrac{\tfrac{\rmtr (\rho\,a\circ b)}{\rmtr (a\circ b)}}
   =-\rmtr (\rho\,a^{1/2}ba^{1/2})\ln\sqbrac{\tfrac{\rmtr (\rho\,a^{1/2}ba^{1/2})}{\rmtr (a^{1/2}ba^{1/2})}}\\
   &=-\rmtr (a\circ\rho\,b)\ln\sqbrac{\tfrac{\rmtr (a\circ\rho\,b)}{\rmtr (ab)}}\hskip 15pc\qedsymbol
\end{align*}
\end{exam}

\begin{exam}{2}  
For $a\in\escript (H)$, $\alpha\in\sscript (H)$ we define the \textit{Holevo operation} \cite{hol82} $\hscript ^{(a,\alpha )}(A)=\rmtr (Aa)\alpha$. Since
\begin{align*}
\rmtr\sqbrac{A\paren{\hscript ^{(a,\alpha )}}^*(B)}&=\rmtr\sqbrac{\hscript ^{(a,\alpha )}(A)B}=\rmtr\sqbrac{\rmtr (Aa)\alpha B}
   =\rmtr (Aa)\rmtr (\alpha B)\\
   &=\rmtr\sqbrac{A\rmtr (\alpha B)a}
\end{align*}
we have $\paren{\hscript ^{(a,\alpha )}}^*(B)=\rmtr (\alpha B)a$. We have $\hscript ^{(a,\alpha )}$ measures $a$ because
\begin{equation*}
\rmtr\sqbrac{\hscript ^{(a,\alpha )}(\rho )}=\rmtr (\rho a)
\end{equation*}
for every $\rho\in\sscript (H)$. We conclude that the $\hscript ^{(a,\alpha )}$ sequential product is
\begin{equation*}
a\circ b=\paren{\hscript ^{(a,\alpha )}}^*(b)=\rmtr (\alpha b)a
\end{equation*}
We also have that
\begin{equation*}
S_{a\circ b}(\rho )=-\rmtr (\alpha b)\rmtr (\rho a)\ln\sqbrac{\tfrac{\rmtr (\rho a)}{\rmtr (a)}}=\rmtr (\alpha b)S_a(\rho )
\end{equation*}
If $a_i\in\escript (H)$, $i=1,2,\ldots ,m$, and we measure $a_i$ with operations $\hscript ^{(a_i,\alpha _i)}$, $i=1,2,\ldots ,m-1$, then
\begin{align*}
a_1\circ a_2\circ\cdots\circ a_m&=a_1\circ (a_2\circ\cdots\circ a_m)=\rmtr (\alpha _1a_2\circ\cdots\circ a_m)a_1\\
   &=\rmtr\sqbrac{\alpha _1\rmtr (\alpha _2a_3\circ\cdots\circ a_m)a_2}a_1\\
   &=\rmtr (\alpha _2a_3\circ\cdots\circ a_m)\rmtr (\alpha _1a_2)a_1\\
   &\vdots\\
   &=\rmtr (\alpha _{m-1}a_m)\rmtr (\alpha _{m-2}a_{m-1})\cdots\rmtr (\alpha _1a_2)a_1
\end{align*}
Moreover, it follows from Corollary~\ref{cor27}(i) that
\begin{equation*}
S_{a_1\circ\cdots\circ a_m}(\rho )=\rmtr (\alpha _{m-1}a_m)\rmtr (\alpha _{m-2}a_{m-1})\cdots\rmtr (\alpha _1a_2)S_{a_1}(\rho )
\end{equation*}
for all $\rho\in\sscript (H)$.\hfill\qedsymbol
\end{exam}

\section{Entropy of Observables and Instruments}  
We now extend our work on entropy of effects to entropy of observables and instruments. An \textit{observable} on $H$ is a finite collection of effects $A=\brac{A_x\colon x\in\Omega _A}$, $A_x\ne 0$, where $\sum\limits _{x\in\Omega _A}A_x=I$ \cite{blm96,hz12,nc00}. The set
$\Omega _A$ is called the \textit{outcome space} of $A$. The effect $A_x$ occurs when a measurement of $A$ results in the outcome $x$. If
$\rho\in\sscript (H)$, then $\rmtr (\rho A_x)$ is the probability that outcome $x$ results from a measurement of $A$ when the system is in state
$\rho$. If $\Delta\subseteq\Omega _A$, then
\begin{equation*}
\Phi _\rho ^A(\Delta )=\sum _{x\in\Delta}\rmtr (\rho A_x)
\end{equation*}
is the probability that $A$ has an outcome in $\Delta$ when the system is in state $\rho$ and $\Phi _\rho ^A$ is called the
\textit{distribution} of $A$. We also use the notation $A(\Delta )=\sum\brac{A_x\colon x\in\Delta}$ so
$\Phi _\rho ^A(\Delta )=\rmtr\sqbrac{\rho A(\Delta )}$ for all $\Delta\subseteq\Omega _A$. In this way, an observable is a
\textit{positive operation-valued measure} (POVM). We say that an observable $A$ is \textit{sharp} if $A_x$ is a projection on $H$ for all
$x\in\Omega _A$ and $A$ is \textit{atomic} if $A_x$ is a one-dimensional projection for all $x\in\Omega _A$.

If $A$ is an observable and $\rho\in\sscript (H)$ the $\rho$-\textit{entropy} of $A$ is $S_A(\rho )=\sum S_{A_x}$ where the sum is over the
$x\in\Omega _A$ such that $\rmtr (\rho A_x)\ne 0$. Then $S_A(\rho )$ is a measure of the information that a measurement of $A$ gives about
$\rho$. The smaller $S_A(\rho )$ is, the more information given. Notice that if $A$ is sharp, then $\rmtr (A_x)=\dim (A_x)$ and if $A$ is atomic, then
\begin{equation*}
S_A(\rho )=-\sum _x\rmtr (\rho A_x)\ln\sqbrac{\rmtr (\rho A_x)}
\end{equation*}
There are two interesting extremes for $S_A(\rho )$. If $\rho$ has spectral decomposition $\rho =\sum\limits _{i=1}^m\lambda _iP_i$ and $A$ is the observable $A=\brac{P_i\colon i=1,2,\ldots ,m}$, then
\begin{equation*}
S_A(\rho )=-\sum _i\rmtr (\rho P_i)\ln\sqbrac{\rmtr (\rho P_i)}=-\sum\lambda _i\ln (\lambda _i)=S(\rho )
\end{equation*}
As we shall see, this gives the minimum entropy (most information). For the completely random state $I/n$ and any observable $A$ we obtain
\begin{align}                
\label{eq31}
S_A(I/n)&=-\sum _x\tfrac{\rmtr (A_x)}{n}\ln\sqbrac{\tfrac{\rmtr (A_x)/n}{\rmtr (A_x)}}=-\tfrac{1}{n}\sum _x\rmtr (A_x)\ln\paren{\tfrac{1}{n}}\notag\\
   &=\tfrac{\ln (n)}{n}\sum _x\rmtr (A_x)=\tfrac{\ln (n)}{n}\,\rmtr (I)=\ln (n)
\end{align}
We shall also see that this gives the maximum entropy (least information).

\begin{thm}    
\label{thm31}
For any observable $A$ and $\rho\in\sscript (H)$ we have
\begin{equation*}
S(\rho )\le S_A(\rho )\le\ln (n)
\end{equation*}
\end{thm}
\begin{proof}
Applying Theorem~\ref{thm21} we obtain
\begin{align*}
S_A(\rho )&=\sum _{x\in\Omega _A}S_{A_x}(\rho )\ge -\sum _{x\in\Omega _A}\sum _i\rmtr (P_iA_x)\lambda _i\ln (\lambda _i)\\
   &=-\sum _i\rmtr\paren{P_i\sum _{x\in\Omega _A}A_x}\lambda _i\ln (\lambda _i)\\
   &=-\sum _i\rmtr (P_i)\lambda _i\ln (\lambda _i)=-\sum _i\lambda _i\ln (\lambda _i)=S(\rho )
\end{align*}
Since $\ln (x)$ is concave and $\rmtr (\rho A_x)>0$, $\sum _x\rmtr (\rho A_x)=1$ we have by Jensen's inequality
\begin{align*}
S_A(\rho )&=\sum _x\rmtr (\rho A_x)\ln\sqbrac{\tfrac{\rmtr (A_x)}{\rmtr (\rho A_x)}}
   \le\ln\sqbrac{\sum _x\rmtr (\rho A_x)\tfrac{\rmtr (A_x)}{\rmtr (\rho A_x)}}\\
   &=\ln\sqbrac{\sum _x\rmtr (A_x)}=\ln\sqbrac{\rmtr (I)}=\ln (n)\qedhere
\end{align*}
\end{proof}

An observable $A$ is \textit{trivial} if $A_x=\lambda _xI$, $0<\lambda _x\le 1$, $\sum\lambda _x=1$.

\begin{cor}    
\label{cor32}
{\rm{(i)}}\enspace $S_A(\rho )=\ln (n)$ if and only if $\rmtr (A_x)\rmtr (\rho A_y)=\rmtr (A_y)\rmtr (\rho A_x)$ for all $x,y\in\Omega _A$.
{\rm{(ii)}}\enspace $A$ is trivial if and only if $S_A(\rho )=\ln (n)$ for all $\rho\in\sscript (H)$.
{\rm{(iii)}}\enspace $\rho =I/n$ if and only if $S_A(\rho )=\ln (n)$ for all observables $A$.
{\rm{(iv)}}\enspace $S(\rho )=\ln (n)$ if and only if $\rho =I/n$.
\end{cor}
\begin{proof}
(i)\enspace This follows from the proof of Theorem~\ref{thm31} because this is the condition for equality in Jensen's inequality.
(ii)\enspace Suppose $A$ is trivial with $A_x=\lambda _xI$. Then for every $\rho\in\sscript (H)$ we have
\begin{equation*}
S_A(\rho )=-\sum _x\rmtr (\rho\lambda _xI)\ln\sqbrac{\tfrac{\rmtr (\rho\lambda _xI)}{\rmtr (\lambda _xI)}}
   =-\sum _x\lambda _x\ln\paren{\tfrac{\lambda _x}{n\lambda _x}}=\ln (n)\sum _x\lambda _x=\ln (n)
\end{equation*}
Conversely, suppose $S_A(\rho )=\ln (n)$ for all $\rho\in\sscript (H)$. By (i) we have that
$\rmtr (A_x)\rmtr (\rho A_y)=\rmtr (A_y)\rmtr (\rho A_x)$ for all $\rho\in\sscript (H)$. It follows that
\begin{equation*}
\elbows{\phi ,A_y\phi}=\elbows{\phi ,A_x\phi}\tfrac{\rmtr (A_y)}{\rmtr (A_x)}
\end{equation*}
for every $\phi\in H$, $\phi\ne 0$. Hence, $A_y= (\rmtr (A_y))/(\rmtr (A_x))A_x$ so that
\begin{equation*}
I=\sum _yA_y=\sum _y\tfrac{\rmtr (A_y)}{\rmtr (A_x)}\,A_x=\tfrac{n}{\rmtr (A_x)}\,A_x
\end{equation*}
We conclude that $A_x=(\rmtr (A_x))/n\,I$ for all $x\in\Omega _A$ so $A$ is trivial.
(iii)\enspace If $\rho =I/n$, we have shown in \eqref{eq31} that $S_A(\rho )=\ln (n)$ for all observables $A$. Conversely, if $S_A(\rho )=\ln (n)$ for every observable $A$, as before, we have $\rmtr (A_x)\rmtr (\rho A_y)=\rmtr (A_y)\rmtr (\rho A_x)$ for every observable $A$. Letting $A_x$ be the observable given by the spectral decomposition $\rho=\sum\lambda _xA_x$ where $A$ is atomic, we conclude that
$\lambda _x=\lambda _y$ for all $x,y\in\Omega _A$. Hence, $\lambda _x=1/n$ and $\rho =\sum (1/n)A_x=I/n$.
(iv)\enspace If $S(\rho )=\ln (n)$, by Theorem~\ref{thm31}, $S_A(\rho )=\ln (n)$ for every observable $A$. Applying (iii), $\rho =I/n$. Conversely, if $\rho =I/n$, then
\begin{equation*}
S(\rho )=-\sum _{i=1}^n\tfrac{1}{n}\,\ln\paren{\tfrac{1}{n}}=-\ln\paren{\tfrac{1}{n}}=\ln (n)\qedhere
\end{equation*}
\end{proof}

We now extend Corollary~\ref{cor27}(ii) and Theorem~\ref{thm28} to observables. If $A^i=\brac{A_x^i\colon x\in\Omega}$ are observables with the same outcome space $\Omega$, $i=1,2,\ldots ,m$, and $0<\lambda _i\le 1$ with $\sum\limits _{i=1}^m\lambda _i=1$, then the observable $A=\brac{A_x\colon x\in\Omega}$ where $A_x=\sum\limits _{i=1}^m\lambda _iA_x^i$ is called a \textit{convex combination} of the $A^i$
\cite{gud20}.

\begin{thm}    
\label{thm33}
{\rm{(i)}}\enspace If $A$ is a convex combination of $A^i$, $i=1,2,\ldots ,m$, then for all $\rho\in\sscript (H)$ we have
\begin{equation*}
S_A(\rho )\ge\sum _{i=1}^m\lambda _iS_{A^i}(\rho )
\end{equation*}
{\rm{(ii)}}\enspace If $0<\lambda _i\le 1$ with $\sum\limits _{i=1}^m\lambda _i=1$, $\rho _i\in\sscript (H)$, $i=1,2,\ldots ,m$, and $A$ is an observable, then
\begin{equation*}
S_A\paren{\sum _i\lambda _i\rho _i}\ge\sum _i\lambda _iS_A(\rho _i)
\end{equation*}
\end{thm}
\begin{proof}
(i)\enspace Applying Corollary~\ref{cor27}(ii) gives
\begin{align*}
S_A(\rho )&=\sum _xS_{A_x}(\rho )=\sum _xS_{\sum\lambda _iA_x^i}(\rho )\ge\sum _x\sum _i\lambda _iS_{A_x^i}(\rho )\\
   &=\sum _i\lambda _i\sum _xS_{A_x^i}(\rho )=\sum _i\lambda _iS_{A^i}(\rho )
\end{align*}
(ii)\enspace Applying Theorem~\ref{thm28} gives
\begin{align*}
S_A\paren{\sum _i\lambda _i\rho _i}&=\sum _xS_{A_x}\paren{\sum _i\lambda _i\rho _i}\ge\sum _x\sum _i\lambda _iS_{A_x}(\rho _i)\\
   &=\sum _i\lambda _i\sum _xS_{A_x}(\rho _i)=\sum _i\lambda _iS_A(\rho _i)\qedhere
\end{align*}
\end{proof}

We say that an observable $B$ is a \textit{coarse-graining} of an observable $A$ if there exists a surjection $f\colon\Omega _A\to\Omega _B$ such that
\begin{equation*}
B_y=\sum\brac{A_x\colon f(x)=y}=A\sqbrac{f^{-1}(y)}
\end{equation*}
for every $y\in\Omega _B$ \cite{gud20,gud22,hz12}.

\begin{thm}    
\label{thm34}
If $B$ is a coarse-graining of $A$, then $S_B(\rho )\ge S_A(\rho )$ for al $\rho\in\sscript (H)$.
\end{thm}
\begin{proof}
Let $B_y=A\sqbrac{f^{-1}(y)}$ for all $y\in\Omega _B$ and let $p_y=\rmtr (\rho B_y)$, $p'_x=\rmtr (\rho A_x)$ for all $y\in\Omega _b$,
$x\in\Omega _A$. Then
\begin{equation*}
p_y=\rmtr\paren{\rho\sum _{f(x)=y}A_x}=\sum _{f(x)=y}\rmtr (\rho A_x)=\sum _{f(x)=y}p'_x
\end{equation*}
Let $V_y=\rmtr (B_y)$, $V'_x=\rmtr (A_x)$ so that
\begin{equation*}
V_y=\rmtr\sum\paren{\sum _{f(x)=y}A_x}=\sum _{f(x)=y}\rmtr (A_x)=\sum _{f(x)=y}V'_x
\end{equation*}
Since $-x\ln (x)$ is concave, we conclude that

\begin{align*}
S_B(\rho )&=-\sum _yp _y\ln\paren{\tfrac{p_y}{V_y}}=-\sum _y\sum _{f(x)=y}p'_x\ln\sqbrac{\frac{\sum _{f(x)=y}p'_x}{V_y}}\\
   &=-\sum _yV_y\paren{\sum _{f(x)=y}\tfrac{p'_xV'_x}{V'_xV_y}}\ln\paren{\sum _{f(x)=y}\tfrac{p'_xV'_x}{V'_xV_y}}\\
   &\ge -\sum _yV_y\sum _{f(x)=y}\tfrac{V'_x}{V_y}\sqbrac{\tfrac{p'_x}{V'_x}\,\ln\paren{\tfrac{p'_x}{V'_x}}}
   =-\sum _y\sum _{f(x)=y}p'_x\ln\paren{\tfrac{p'_x}{V'_x}}\\
   &=-\sum _xp'_x\ln\paren{\tfrac{p'_x}{V'_x}}=S_A(\rho )\qedhere
\end{align*}
\end{proof}

The equality condition for Jensen's inequality gives the following.

\begin{cor}    
\label{cor35}
An observable $A$ possesses a coarse-graining $B_y=A\sqbrac{f^{-1}(y)}$ with $S_B(\rho )=S_A(\rho )$ for all $\rho\in\sscript (H)$ if and only if for every $x_1,x_2\in\Omega _A$ with $f(x_1)=f(x_2)$ we have
\begin{equation*}
\rmtr (A_{x_2})\rmtr (\rho A_{x_1})=\rmtr (A_{x_1})\rmtr (\rho A_{x_2})
\end{equation*}
\end{cor}

A trace preserving operation is called a \textit{channel}. An \textit{instrument} on $H$ is a finite collection of operations
$\iscript =\brac{\iscript _x\colon x\in\Omega}$ such that $\sum _{x\in\Omega _\iscript}\iscript _x$ is a channel \cite{blm96,hz12,nc00}. We call
$\Omega _\iscript$ the \textit{outcome space} for $\iscript$. If $\iscript$ is an instrument, there exists a unique observable $A$ such that
$\rmtr (\rho A_x)=\rmtr\sqbrac{\iscript _x(\rho )}$ for all $x\in\Omega _A=\Omega _\iscript$, $\rho\in\sscript (H)$ and we say that $\iscript$
\textit{measures} $A$. Although an instrument measures a unique observable, an observable is measured by many instruments For example, if $A$ is an observable, the corresponding \textit{\L\"uders instrument} \cite{lud51} is defined by
\begin{equation*}
\lscript _x^A(B)=A_x^{1/2}BA_x^{1/2}
\end{equation*}
for all $B\in\lscript (H)$. Then $\lscript ^A$ is an instrument because
\begin{align*}
\rmtr\sqbrac{\sum _x\lscript _x^A(B)}&=\sum _x\rmtr\sqbrac{\lscript _x^A(B)}=\sum _x\rmtr (A_x^{1/2}BA_x^{1/2})=\sum _x\rmtr (A_xB)\\
   &=\rmtr\paren{\sum _xA_xB}=\rmtr (IB)=\rmtr (B)
\end{align*}
for all $B\in\lscript (H)$. Moreover, $\lscript ^A$ measures $A$ because
\begin{equation*}
\rmtr\sqbrac{\lscript _x^A(\rho )}=\rmtr (A_x^{1/2}\rho A_x^{1/2})=\rmtr (\rho A_x)
\end{equation*}
for all $\rho\in\sscript (H)$. Of course, this is related to Example~1. Corresponding to Example~2, we have a \textit{Holevo instrument}
$\hscript ^{(A,\alpha )}$ where $\alpha _x\in\sscript (H)$, $x\in\Omega _A$ and
\begin{equation*}
\hscript _x^{(A,\alpha )}(B)=\rmtr (BA_x)\alpha _x
\end{equation*}
for all $B\in\lscript (H)$ \cite{hol82}. To show that $\hscript ^{(A,\alpha )}$ is an instrument we have
\begin{align*}
\rmtr\sqbrac{\sum _x\hscript _x^{(A,\alpha )}(B)}&=\sum _x\rmtr\sqbrac{\hscript _x^{(A,\alpha )}(B)}
   =\sum _x\rmtr\sqbrac{\rmtr (BA_x)\alpha _x}\\
   &=\sum _x\rmtr (BA_x)=\rmtr\paren{B\sum _xA_x}=\rmtr (B)
\end{align*}
Moreover, $\hscript ^{(A,\alpha )}$ measures $A$ because
\begin{equation*}
\rmtr\sqbrac{\hscript _x^{A,\alpha}(\rho )}=\rmtr\sqbrac{(\rho A_x)\alpha _x}=\rmtr (\rho A_x)\rmtr (\alpha _x)=\rmtr (\rho A_x)
\end{equation*}

Let $A,B$ be observables and let $\iscript$ be an instrument that measures $A$. We define the $\iscript$-\textit{sequential} product
$A\circ B$ \cite{gud20,gud21} by $\Omega _{A\circ B}=\Omega _A\times\Omega _B$ and
\begin{equation*}
A\circ B_{(x,y)}=\iscript _x^*(B_y)=A_x\circ B_y
\end{equation*}
Defining $f\colon\Omega _{A\circ B}\to\Omega _A$ by $f(x,y)=x$,we obtain
\begin{equation*}
A\circ B\sqbrac{f^{-1}(x)}=\sum _{f(x,y)=x}A_x\circ B_y=\sum _{y\in\Omega _B}\iscript _x^*(B_y)=\iscript _\alpha ^*(I)=A_x
\end{equation*}
We conclude that $A$ is a coarse-graining of $A\circ B$. Applying Theorem~\ref{thm34} we obtain the following.

\begin{cor}    
\label{cor36}
If $A,B$ are observables, the $S_{A\circ B}(\rho )\le S_A(\rho )$ for all $\rho\in\sscript (H)$. Equality $S_{A\circ B}(\rho )=S_A(\rho )$ holds if and only if for every $x\in\Omega _A$, $y_1,y_2\in\Omega _B$ we have
\begin{equation*}
\tfrac{\rmtr (\rho A_x\circ B_{y_1})}{\rmtr (A_x\circ B_{y_1})}\,\ln\sqbrac{\tfrac{\rmtr (\rho A_x\circ B_{y_1})}{\rmtr (A_x\circ B_{y_1})}}
   =\tfrac{\rmtr (\rho A_x\circ B_{y_2})}{\rmtr (A_x\circ B_{y_2})}\,\ln\sqbrac{\tfrac{\rmtr (\rho A_x\circ B_{y_2})}{\rmtr (A_x\circ B_{y_2})}}
\end{equation*}
\end{cor}

Extending this work to more than two observables, let $\iscript ^1,\iscript ^2,\ldots ,\iscript ^{m-1}$ be instruments that measure the observables $A^1,A^2,\ldots ,A^{m-1}$, respectively. If $A^m$ is another observable, we have that
\begin{equation*}
(A^1\circ A^2\circ\cdots\circ A^m)_{(x_1,x_2,\ldots ,x_m)}=(\iscript _{x_1}^1)^*(\iscript _{x_2}^2)^*
    \cdots (\iscript _{x_{m-1}}^{m-1})^*(A_{x_m}^m)
\end{equation*}
The next result follows from Corollary~\ref{cor36}.

\begin{cor}    
\label{cor37}
If $A^1,A^2,\ldots ,A^m$ are observables, then
\begin{equation*}
S_{A^1\circ A^2\circ\cdots\circ A^m}(\rho )\le S_{A^1\circ A^2\circ\cdots\circ A^{m-1}}(\rho )
\end{equation*}
for all $\rho\in\sscript (H)$.
\end{cor}

If $\iscript$ is an instrument, let $A$ be the unique observable that $\iscript$ measures so $\rmtr\sqbrac{\iscript _x(\rho )}=\rmtr (\rho A_x)$ for all $x\in\Omega _\iscript$ and $\rho\in\sscript (H)$. We define the $\rho$-\textit{entropy} of $\iscript$ as $S_\iscript (\rho )=S_A(\rho )$. Since
$A_x=\iscript _x^*(I)$ we have
\begin{equation*}
\rmtr (A_x)=\rmtr\sqbrac{\iscript _x^*(I)}=\rmtr\sqbrac{\iscript _x(I)}
\end{equation*}
Hence,
\begin{equation*}
S_\iscript (\rho )=S_A(\rho )=-\sum _x\rmtr (\rho A_x)\ln\sqbrac{\tfrac{\rmtr (\rho A_x)}{\rmtr (A_x)}}
   =-\sum _x\rmtr\sqbrac{\iscript _x(\rho )}\ln\brac{\tfrac{\rmtr\sqbrac{\iscript _x(\rho )}}{\rmtr\sqbrac{\iscript _x(I)}}}
\end{equation*}
Now let $\iscript ^1,\iscript ^2,\ldots ,\iscript ^m$ be instruments and let $A^1,A^2,\ldots ,A^m$ be the unique observables they measure, respectively. Denoting the composition of two instruments $\iscript,\jscript$ by $\iscript\circ\jscript$ we have
\begin{align*}
\rmtr\sqbrac{\iscript _{x_m}^m\circ\iscript _{x_{m-1}}^{m-1}\circ\cdots\circ\iscript _{x_1}^1(\rho )}
   &=\rmtr\sqbrac{\rho (\iscript _{x_1}^1)^*(\iscript _{x_2}^1)^*\cdots (\iscript _{x_m}^m)^*(I)}\\
   &=\rmtr (\rho A_{x_1}^1\circ A_{x_2}^2\circ\cdots\circ A_{x_m}^m)
\end{align*}
Hence, the observable measured by $\iscript ^m\circ\iscript ^{m-1}\circ\cdots\circ\iscript ^1$ is $A^1\circ A^2\circ\cdots\circ A^m$. It follows that
\begin{equation*}
S_{\iscript ^m\circ\iscript ^{m-1}\circ\cdots\circ\iscript ^1}(\rho )=S_{A^1\circ A^2\circ\cdots\circ A^m}(\rho )
\end{equation*}
We conclude that Theorem~1, 2 and 3 \cite{st22} follow from our results. Moreover, our proofs are simpler since they come from the more basic concept of $\rho$-entropy for effects.

Let $A,B$ be observables on $H$ and let $\iscript$ be an instrument that measures $A$. The corresponding sequential product becomes
\begin{equation*}
(A\circ B)_{(x,y)}=\iscript _x^*(B_y)=A_x\circ B_y
\end{equation*}
The $\rho$-entropy of $A\circ B$ has the form
\begin{align*}
S_{A\circ B}(\rho )
   &=-\sum _{x,y}\rmtr\sqbrac{\rho (A\circ B)_{(x,y)}}\ln\brac{\tfrac{\rmtr\sqbrac{\rho (A\circ B)_{(x,y)}}}{\rmtr\sqbrac{(A\circ B)_{(x,y)}}}}\\
   &=-\sum _{x,y}\rmtr\sqbrac{\rho\iscript _x^*(B_y)}\ln\brac{\tfrac{\rmtr\sqbrac{\rho\iscript _x^*(B_y)}}{\rmtr\sqbrac{\iscript _x^*(B_y)}}}\\
   &=-\sum _{x,y}\rmtr\sqbrac{\iscript _x(\rho )B_y}\ln\brac{\tfrac{\sqbrac{\iscript _x(\rho )B_y}}{\rmtr\sqbrac{\iscript _x(I)B_y}}}
\end{align*}
If $\iscript ^A$ is the L\"uders instrument $\iscript _x^A(\rho )=A_x^{1/2}\rho A_x^{1/2}$ we have $(A\circ B)_{(x,y)}=A_x^{1/2}B_yA_x^{1/2}$ and
\begin{equation*}
S_{A\circ B}(\rho )=-\sum _{x,y}\rmtr (A_x^{1/2}\rho A_x^{1/2}B_y)\ln\sqbrac{\frac{\rmtr (A_x^{1/2}\rho A_x^{1/2}B_y)}{\rmtr (A_xB_y)}}
\end{equation*}
If $\hscript ^{(A,\alpha )}$ is the Holevo instrument $\hscript _x^{(A,\alpha )}(\rho )=\rmtr (\rho A_x)\alpha _x$, $\alpha _x\in\sscript (H)$ we obtain
\begin{align*}
S_{A\circ B}(\rho )
   &=-\sum _{x,y}\rmtr (\rho A_x)\rmtr (\alpha _xB_y)\ln\sqbrac{\tfrac{\rmtr (\rho A_x)\rmtr (\alpha _xB_y)}{\rmtr (A_x)\rmtr (\alpha _xB_y)}}\\
   &=-\sum _{x,y}\rmtr (\rho A_x)\rmtr (\alpha _xB_y)\ln\sqbrac{\tfrac{\rmtr (\rho A_x)}{\rmtr (A_x)}}\\
   &=-\sum _x\rmtr (\rho A_x)\ln\sqbrac{\tfrac{\rmtr (\rho A_x)}{\rmtr (A_x)}}=S_A(\rho )
\end{align*}
This also follows from Corollary~\ref{cor36} because
\begin{equation*}
\tfrac{\rmtr (\rho A_x\circ B_y)}{\rmtr (A_x\circ B_y)}=\tfrac{\rmtr (\alpha _xB_y)\rmtr (\rho A_x)}{\rmtr (\alpha _xB_y)\rmtr (A_x)}
   =\tfrac{(\rho A_x)}{\rmtr (A_x)}
\end{equation*}

If $A$ is an observable on $H$ and $B$ is an observable on $K$ we form the \textit{tensor product observable} $A\otimes B$ on $H\otimes K$ given by $(A\otimes B)_{(x,y)}=A_x\otimes B_y$ where $\Omega _{A\otimes B}=\Omega _A\times\Omega _B$ \cite{gud20}.

\begin{lem}    
\label{lem38}
If $\rho _1\in\sscript (H)$, $\rho _2\in\sscript (K)$, then
\begin{equation*}
S_{A\circ B}(\rho _1\otimes\rho _2)=S_A(\rho _1)+S_B(\rho _2)
\end{equation*}
\end{lem}
\begin{proof}
From the definition of $A\otimes B$ we obtain
\begin{align*}
S_{A\otimes B}(\rho _1\otimes\rho _2)&=-\sum _{x,y}\rmtr (\rho _1\otimes\rho _2A_x\otimes B_y)\ln
   \sqbrac{\tfrac{\rmtr (\rho _1\otimes\rho _2A_x\otimes B_y)}{\rmtr (A_x\otimes B_y)}}\\
   &=-\sum _{x,y}\rmtr (\rho _1A_x)\rmtr (\rho _2B_y)\ln\sqbrac{\tfrac{\rmtr (\rho _1A_x)\rmtr (\rho _2B_y)}{\rmtr (A_x)\rmtr (B_y)}}\\
   &=-\sum _{x,y}\rmtr (\rho _1A_x)\rmtr (\rho _2B_y)\ln\sqbrac{\tfrac{\rmtr (\rho _1A_x)}{\rmtr (A_x)}}\\
   &\quad -\sum _{x,y}\rmtr (\rho _1A_x)\rmtr (\rho _2B_y)\ln\sqbrac{\tfrac{\rmtr (\rho _2B_y)}{\rmtr (B_y)}}\\
   &=-\sum _x\rmtr (\rho _1A_x)\ln\sqbrac{\tfrac{\rmtr (\rho _1A_x)}{\rmtr (A_x)}}
   -\sum _y\rmtr (\rho _2B_y)\ln\sqbrac{\tfrac{\rmtr (\rho _2B_y)}{\rmtr (B_y)}}\\
   &=S_A(\rho _1)+S_B(\rho _2)\qedhere
\end{align*}
\end{proof}
We conclude that $A$ gives more information about $\rho _1$ than $A$ and $B$ give about $\rho _1\otimes\rho _2$ and similarly for $B$.

A \textit{measurement model} \cite{blm96,hz12,nc00} is a 5-tuple $\mscript =(H,K,\nu ,\sigma ,P)$ where $H$ is the \textit{system} Hilbert space, $K$ is the \textit{probe} Hilbert space, $\nu$ is the \textit{interaction} channel, $\sigma\in\sscript (K)$ is the initial \textit{probe state} and $P$ is the \textit{probe observable} on $K$. We interpret $\mscript$ as an apparatus that is employed to measure an instrument and hence an observable. In fact, $\mscript$ measures the unique instrument $\iscript$ on $H$ given by
\begin{equation*}
\iscript _x(\rho )=\rmtr _K\sqbrac{\nu (\rho\otimes\sigma )(I\otimes P_x)}
\end{equation*}
In this way, a state $\rho\in\sscript (H)$ is input into the apparatus and combined with the initial state $\sigma$ of the probe system. The channel $\nu$ interacts the two states and a measurement of the probe $P$ is performed resulting in outcome $x$. The outcome state is reduced to $H$ by applying the partial trace over $K$. Now $\iscript$ measures an unique observable $A$ on $H$ that satisfies
\begin{equation}                
\label{eq32}
\rmtr (\rho A_x)=\rmtr\sqbrac{\iscript _x(\rho )}=\rmtr\sqbrac{\nu (\rho\otimes\sigma )(I\otimes P_x)}
\end{equation}
The $\rho$-entropy of $\iscript$ becomes
\begin{equation*}
S_\iscript (\rho )=S_A(\rho )=-\sum _x\rmtr (\rho A_x)\ln\sqbrac{\tfrac{\rmtr (\rho A_x)}{\rmtr (A_x)}}
\end{equation*}
where $\rmtr (\rho A_x)$ is given by \eqref{eq32}. Of course, $S_\iscript (\rho )=S_A(\rho )$ gives the amount of information that a measurement by $\mscript$ provides about $\rho$. A closely related concept is the observable $I\otimes P$ and
$S_{I\otimes P}\sqbrac{\nu (\rho\otimes\sigma )}$ also provides the amount of information that a measurement $\mscript$ provides about
$\rho$. It follows from \eqref{eq32} that the distribution of $A$ in the state $\rho$ equals the distribution of $I\otimes P$ in the state
$\nu (\rho\otimes\sigma )$. We now compare $S_A(\rho )$ and $S_{I\otimes P}\sqbrac{\nu (\rho\otimes\sigma )}$. Applying \eqref{eq32} gives
\begin{align*}
S_{I\otimes P}&\sqbrac{\nu (\rho\otimes\sigma )}\\
    &=-\sum _x\rmtr\sqbrac{\nu (\rho\otimes\sigma )(I\otimes P_x)}
    \ln\brac{\tfrac{\rmtr\sqbrac{\nu (\rho\otimes\sigma )(I\otimes P_x)}}{\rmtr (I\otimes P_x)}}\\
    &=-\sum _x\rmtr (\rho A_x)\ln\sqbrac{\tfrac{\rmtr (\rho A_x)}{n\rmtr (P_x)}}
    =-\sum _x\rmtr (\rho A_x)\ln\sqbrac{\tfrac{\rmtr (A_x)}{n\rmtr (P_x)}\,\tfrac{\rmtr (\rho A_x)}{\rmtr (A_x)}}\\
    &=-\sum _x\rmtr (\rho A_x)\ln\sqbrac{\tfrac{\rmtr (\rho A_x)}{\rmtr (A_x)}}-\sum\rmtr (\rho A_x)\ln\sqbrac{\tfrac{\rmtr (A_x)}{n\rmtr (P_x)}}\\
    &=S_A(\rho )-\sum _x\rmtr (\rho A_x)\ln\sqbrac{\tfrac{\rmtr (A_x)}{n\rmtr (P_x)}}
\end{align*}
It follows that $S_A(\rho )\le S_{I\otimes P}\sqbrac{\nu (\rho\otimes\sigma )}$ if and only if
\begin{equation}                
\label{eq33}
\sum _x\rmtr (\rho A_x)\ln\sqbrac{\tfrac{\rmtr (A_x)}{n\rmtr (P_x)}}\le 0
\end{equation}
Now \eqref{eq33} may or may not hold depending on $A$, $\rho$ and $P$. In many cases, $P$ is atomic \cite{blm96,hz12} and then
\begin{equation*}
\ln\sqbrac{\tfrac{\rmtr (A_x)}{n\rmtr (P_x)}}=\ln\sqbrac{\tfrac{\rmtr (A_x)}{n}}<0
\end{equation*}
so $S_A(\rho)\le S_{I\otimes P}\sqbrac{\nu (\rho\otimes\sigma )}$ for all $\rho\in\sscript (H)$. Also, \eqref{eq33} holds if $P$ is sharp.

\end{document}